\documentclass[compsoc]{IEEEtran}
\usepackage{amsthm}
\usepackage{amsmath,amssymb,amsfonts}

\usepackage{cite}

\usepackage{longtable}
\usepackage{siunitx}
\usepackage{color}
\usepackage{bm}
\usepackage[colorlinks,allcolors=blue]{hyperref}
\usepackage{caption}

\usepackage{longtable}

\usepackage{booktabs}
\usepackage{longtable}
\usepackage{array}
\usepackage{multirow}
\usepackage{wrapfig}
\usepackage{float}
\usepackage{colortbl}
\usepackage{pdflscape}
\usepackage{tabu}
\usepackage{threeparttable}
\usepackage{threeparttablex}
\usepackage[normalem]{ulem}
\usepackage{makecell}
\usepackage{longtable}

\def\BibTeX{{\rm B\kern-.05em{\sc i\kern-.025em b}\kern-.08em
    T\kern-.1667em\lower.7ex\hbox{E}\kern-.125emX}}
  
\newcommand{\R}{\mathbb{R}}

\newcommand{\Var}{{\rm Var}}
\newcommand{\Vartildehat}{\hat{\tilde{{\rm Var}}}}

\newcommand{\Covtildehat}{\hat{\tilde{{\rm Cov}}}}

\newcommand{\diag}{{\rm diag}}

\newcommand{\cind}{\overset{\text{c}}{\ind}}
\newcommand{\epsind}{\overset{\varepsilon}{\ind}}
\newcommand{\ind}{\perp \!\!\! \perp}

\newcommand{\Xb}{\bm{X}}
\newcommand{\xb}{\bm{x}}

\newcommand{\vbtilde}{\tilde{\bm{v}}}
\newcommand{\vbtildehat}{\hat{\tilde{\bm{v}}}}
\newcommand{\vtildehat}{\hat{\tilde{v}}}
\newcommand{\vtilde}{\tilde{v}}

\newcommand{\ubtildehat}{\hat{\tilde{\bm{u}}}}

\newcommand{\Ubtildehat}{\hat{\tilde{\bm{U}}}}
\newcommand{\Omegabtildehat}{\hat{\tilde{\bm{\Omega}}}}
\newcommand{\omegatildehat}{\hat{\tilde{\omega}}}

\newcommand{\Vbtildehat}{\hat{\tilde{\bm{V}}}}

\newcommand{\xib}{\bm{\xi}}

\newcommand{\Sigmab}{\bm{\Sigma}}

\newcommand{\Sigmabtilde}{\tilde{\bm{\Sigma}}}
\newcommand{\Sigmabtildehat}{\hat{\tilde{\bm{\Sigma}}}}
\newcommand{\sigmatilde}{\tilde{\sigma}}
\newcommand{\sigmatildehat}{\hat{\tilde{\sigma}}}

\newcommand{\lambdatilde}{\tilde{\lambda}}
\newcommand{\lambdatildehat}{\hat{\tilde{\lambda}}}

\newcommand{\Lambdabtildehat}{\hat{\tilde{\bm{\Lambda}}}}

\newcommand{\jN}{{j|N}}

\newcommand{\Beta}{\mathrm{B}}

\newcommand{\Etab}{\bm{\mathrm{H}}}
\newcommand{\etab}{\bm{\eta}}

\newcommand{\Epsilonb}{\bm{\mathrm{E}}}
\newcommand{\epsilonb}{\bm{\varepsilon}}
\newcommand{\Rho}{\mathrm{P}}
\newcommand{\Rhob}{\bm{\Rho}}
\newcommand{\Rhobtilde}{\tilde{\bm{\Rho}}}
\newcommand{\Rhobtildehat}{\hat{\tilde{\bm{\Rho}}}}

\newcommand{\bD}{\mathcal{D}}
\newcommand{\bDelta}{\mathit{\Delta}}

\newtheorem{Definition}{Definition}[section]

\newtheorem{Theorem}{Theorem}[section]

\newtheorem{Corollary}{Corollary}[section]

\newtheorem{Algorithm}{Algorithm}[section]

\newcommand{\refBauer}{Theorem 4.2 in \cite{BD20}}
\newcommand{\refDavis}{Corollary 1 in \cite{YW15}}

\numberwithin{equation}{section} 
\begin{document}

\title{Correlation Based Principal Loading Analysis}

\author{\IEEEauthorblockN{Jan O. Bauer}\\
\IEEEauthorblockA{\textit{Baden-Wuerttemberg Cooperative State University Mannheim} \\
Mannheim, Germany \\
jan.bauer@dhbw-mannheim.de}
}
\maketitle
\begin{abstract}
Principal loading analysis is a dimension reduction method that discards variables which have only a small distorting effect on the covariance matrix. We complement principal loading analysis and propose to rather use a mix of both, the correlation and covariance matrix instead. Further, we suggest to use rescaled eigenvectors and provide updated algorithms for all proposed changes.
\end{abstract}


\section{Introduction}
Principal loading analysis (PLA) is a tool developed by \cite{BD20} to reduce dimensions. Their method chooses a subset of observed variables by discarding the other variables based on the impact of the eigenvectors on the covariance matrix. While the method itself is new, some parts of PLA correspond with principal component analysis (PCA) which is a popular dimension reduction technique first formulated by \cite{PE01} and \cite{HO33}. Despite their intersection however, the outcome is as different as can be since PCA yields a reduced set of variables by transforming the original variables while PLA selects a subset of the original variables. Nonetheless, PLA partially adopts established concepts from PCA.\\
PLA is originally based on the covariance matrix. However, we feel that this comes at a price due to the lack of scale invariance of the covariance matrix and, hence, we propose rather to use both, the covariance and the correlation matrix. Each for different steps of PLA. Therefore, our contribution is an adjusted method of PLA. Further, we suggest to rescale the eigenvectors and we provide simulations to find optimal cut-off values. \\
This article is organized as follows: Section~\ref{s:setup} provides notation needed for the remainder of this work. In Section~\ref{s:PLA}, we recap PLA based on the covariance matrix. The focus of this article is Section~\ref{s:CovPLA} where we elaborate the issue regarding the usage of the covariance matrix. Our main result is summarized in Corollary~\ref{co:derivative}. In Section~\ref{s:CorrPLA} and Section~\ref{s:rescaledeigenvectors} we provide updated algorithms based on the correlation matrix while we implement rescaled eigenvectors in the latter algorithm. Section~\ref{s:simulation} contains our conclusion for optimal threshold values from the simulation studies while we also briefly cover simulation difficulties regarding PLA. Finally, we summarize our work in Section~\ref{s:conclusion}.

\section{Setup}
\label{s:setup}
\noindent
Let $\xb = \begin{pmatrix} \xb_1 & \cdots & \xb_M \end{pmatrix}$ be a $N\times M$ sample containing $n \in\{ 1,\ldots,N\}$ observations of a random vector $\Xb = (X_1 \;\cdots\; X_M)$ with covariance matrix $\Sigmab = (\sigma_{i,j})$ for defined indices $i,j\in\{1,\ldots,M\}$ throughout this work. We consider the case when the covariance matrix is slightly perturbed by a sparse matrix $\Epsilonb = (\varepsilon_{i,j})$ such that
$$\Sigmabtilde \equiv \Sigmab + \Epsilonb  \; . $$
$\Epsilonb$ is a technical construction and contains small components we want to extract from $\Sigmab$. Hence, $\varepsilon_{i,j}\neq0\Rightarrow\sigma_{i,j}=0$. The sample counterpart of $\Sigmabtilde$ is of the form
$$ \Sigmabtildehat \equiv (\sigmatildehat_{i,j}) \equiv \Sigmab + \Epsilonb  + \Etab_N $$
where $\Etab_N = (\eta_{i,j|N})$ is a perturbation in the form of a random noise matrix. The noise is due to having only a finite number of observations in the sample. The correlation matrix and the sample correlation matrix are denoted by $\Rhobtilde$ and $\Rhobtildehat$ respectively. We consider the eigendecomposition of $\Sigmabtildehat$ to be given by 
\begin{equation}
\label{eq:eigendec}
\Sigmabtildehat \equiv \Vbtildehat\Lambdabtildehat\Vbtildehat^\top
\end{equation}
with $\Vbtildehat^\top\Vbtildehat=\bm{I}$, $\Lambdabtildehat = \diag(\lambdatildehat_1,\ldots,\lambdatildehat_M)$ and $\lambdatildehat_1\geq\ldots\geq\lambdatildehat_M$. The eigenvectors $\Vbtildehat = \begin{pmatrix} \vbtildehat_1  & \cdots & \vbtildehat_M\end{pmatrix}$ are ordered according to the respective eigenvalues. The eigendecomposition of $\Sigmabtilde$ is denoted analogously. The eigendecomposition of the sample correlation matrix is given by $\Rhobtildehat = \Ubtildehat \Omegabtildehat \Ubtildehat^\top$ where $\Ubtildehat^\top\Ubtildehat=\bm{I}$ and $\Omegabtildehat = \diag(\omegatildehat_1,\ldots,\omegatildehat_M)$.\\
When two blocks of random variables $X_{i_1},\ldots,X_{i_I}$ and $X_{j_1},\ldots,X_{j_J}$ are uncorrelated we write $(X_{i_1},\ldots,X_{i_I})\cind (X_{j_1},\ldots,X_{j_J})$. Further, we recap the definition of $\varepsilon$-uncorrelatedness provided by \cite{BD20}:
\begin{Definition}
We say that two blocks of random variables $X_{i_1},\ldots,X_{i_I}$ and $X_{j_1},\ldots,X_{j_J}$ are $\varepsilon$-uncorrelated if $ \sigmatilde_{i_{\bar{i}},j_{\bar{j}}} = \varepsilon_{i_{\bar{i}},j_{\bar{j}}}$ for $\bar{i}\in\{1,\ldots,I\}$ and $\bar{j}\in\{1,\ldots,J\}$. We then write $(X_{i_1},\ldots,X_{i_I})\epsind (X_{j_1},\ldots,X_{j_J})$. 
\end{Definition}
\noindent
For practical purposes we define $\bD \subset \{1,\ldots,M\}$ such that $|\bD| = M^\ast$  where $1\leq M^\ast < M$. $\bD$ will contain the indices of the variables $\{X_d\} \equiv \{X_d\}_{d\in\bD}$ we consider to discard and we introduce the shortcut $d^c\notin \bD$ when referring to the case that $d^c\in\{1,\ldots,M\} \backslash \bD$. In the same spirit, $\bDelta$ with $|\bDelta| = M^\ast$ will be used to index eigenvectors with respective eigenvalues linked to the $\{X_d\}$ and $\delta^c\notin\bDelta$ refers to $\delta^c\in\{1,\ldots,M\} \backslash \bDelta$. Further, the elements of any $(M\times1)$ vector $\xib$ are denoted by $\xib =\begin{pmatrix} \xi^{(1)} & \cdots &\xi^{(M)}\end{pmatrix}^\top$.
\section{Principal Loading Analysis}
\label{s:PLA}
\noindent
PLA is a tool for dimension reduction where a subset of existing variables is selected while the other variables are discarded. The intuition is that blocks of variables are discarded which distort the covariance matrix only slightly. It will turn out that those blocks are specified by $\Sigmabtildehat$. Firstly, we recap PLA and deepen the understanding of the explained variance in subsection \ref{ss:EV} and afterwards restate the procedure of PLA in subsection \ref{ss:PLAAlg}.

\subsection{Contribution to the Explained Variance}
\label{ss:EV}
We start to recap the method by assuming that only a single block, consisting of the variables $\{X_d\}$ is discarded. Those variables distort the covariance matrix only by a little if the $M^\ast$ rows indexed by $d\in\bD$ of $M-M^\ast$ eigenvectors are small in absolute terms and if the contribution of those eigenvectors to the explained variance is large, hence if the contribution of the other eigenvectors is small.\\
We assume that the eigenvectors $\{\vbtildehat_{\delta^c}\}_{d^c\notin\bDelta}$ are the eigenvectors with small absolute elements $\{  |\vtildehat_{\delta^c}^{(d)}| \}_{d\in \bD,\delta^c\notin\bDelta}$. Consequently, $\{\vbtildehat_{\delta}\}_{\delta\in\bDelta}$ contain small absolute elements $\{  |\vtildehat_{\delta}^{(d^c)}| \leq \tau \}_{\delta\in\bDelta,d^c\notin\bD}$, where $\tau$ is a chosen threshold, since the eigenvectors are orthonormal due to the symmetry of $\Sigmabtildehat$. The percent contribution of the $\{\vbtildehat_{\delta}\}_{\delta\in\bDelta}$ to the explained variance of $\Sigmabtildehat$ is then given by
\begin{equation}
\label{eq:ExpVarPLA1}
\left(  \sum\limits_{i} \lambdatildehat_i  \right)^{-1} \left( \sum\limits_{\delta\in\bDelta} \lambdatildehat_{\delta} \sum\limits_{d\in\bD} (\vtildehat_{\delta}^{(d)})^2 +    \sum\limits_{\delta^c\notin\bDelta} \lambdatildehat_{{\delta^c}} \sum\limits_{d\in\bD} (\vtildehat_{\delta^c}^{(d)})^2 \right) 
\end{equation}
which equals the contribution of the block containing the $\{X_{d}\}$. The intuition of (\ref{eq:ExpVarPLA1}) is as follows: considering $\tau=0$, the expression reduces to
\begin{equation}
\label{eq:ExpVarPCA}
\left(  \sum\limits_{i} \lambdatildehat_i  \right)^{-1} \left( \sum\limits_{\delta\in\bDelta} \lambdatildehat_{\delta} \right) 
\end{equation}
since $\vtildehat_{\delta^c}^{(d)} = \vtildehat_{\delta}^{(d^c)}  = 0$ for all $d\in\bD$, $d^c\notin\bD$, $\delta\in\bDelta$ and $\delta^c\notin\bDelta$. Hence, it also holds that $1 = \Vert \vbtildehat_\delta \Vert_2^2 =\sum_{d\in\bD} (\vtildehat_{\delta}^{(d)})^2$. (\ref{eq:ExpVarPCA}) then is the percent contribution of the linear combination of $\{X_{d}\}$ to the explained variance since
$$\Var\left(\sum_{d\in\bD} \vtilde_\delta^{(d)} X_d \right) = \Var( \vbtilde_\delta^\top \Xb ) = \vbtilde_\delta^\top \Sigmabtilde \vbtilde_\delta = \lambdatilde_\delta $$  is the population contribution of the $\{X_{d}\}$ into the direction of $\vbtilde_\delta$ and since the population contribution of the $\{X_{d}\}$ in all directions is therefore given by $\sum_{\delta\in\bDelta} \lambdatilde_{\delta} $. Consequently, the sample counterpart is $\sum_{\delta\in\bDelta} \lambdatildehat_{\delta} $ which we divide by the overall explained variance $\sum_{i} \lambdatildehat_i $ in order to obtain a percent contribution. Note that the elements in (\ref{eq:ExpVarPLA1}) are squared since the eigenvectors are normed one and therefore each squared element, say the $i$-th element, can be interpreted as the percent contribution of the corresponding random variable $X_i$ into the direction of the eigenvector. \cite{BD20}\\
Since usually $\tau\neq 0$ with $\tau$ small however, (\ref{eq:ExpVarPCA}) serves as a fairly good approximation of (\ref{eq:ExpVarPLA1}). This is due to the sparseness of $\Etab$ which bounds $\max_j |\lambdatilde_j - \lambda_j| \leq \Vert \Etab \Vert_F$ as shown by \cite{BD20} and further because $\vtildehat_{\delta^c}^{(d)} \approx \vtildehat_{\delta}^{(d^c)}  \approx 0$.

\subsection{PLA Algorithm}
\label{ss:PLAAlg}
Variables cause small components in the eigenvectors when the variables are arranged in blocks\footnote{Note that we can assume that the covariance matrix behaves in this convenient way because we can always obtain this structure using a permutation matrix.} such that
$$\underbrace{ (X_1,\ldots,X_{M_1}) }_{\kappa_1\text{-many}}  \epsind \ldots \epsind \underbrace{  (X_{M_{L-1}+1},\ldots,X_{M_L})  }_{\kappa_L\text{-many}} $$
which can be denoted as
\begin{equation}
\label{eq:blockcovariancematrix}
\underbrace{\Sigmabtilde_1}_{\kappa_1\times\kappa_1}  \epsind   \ldots \epsind \underbrace{\Sigmabtilde_L}_{\kappa_L\times\kappa_L}
\end{equation}
to emphasize the block structure. This follows since the population eigenvectors of $\Sigmab$ are then of shape
\begin{equation}
\label{eq:blockeigenstructure}
{\footnotesize
\underbrace{\begin{pmatrix}  \bm{\ast}_{\kappa_{1}} \\ \bm{0} \end{pmatrix},\ldots,\begin{pmatrix}  \bm{\ast}_{\kappa_{1}} \\ \bm{0} \end{pmatrix}}_{\kappa_{1}\text{-many}}  ,    
\underbrace{\begin{pmatrix} \bm{0} \\ \bm{\ast}_{\kappa_{2}} \\ \bm{0} \end{pmatrix},\ldots,\begin{pmatrix} \bm{0} \\ \bm{\ast}_{\kappa_{2}}\\ \bm{0} \end{pmatrix}}_{\kappa_{2}\text{-many}}  , \ldots , 
\underbrace{\begin{pmatrix} \bm{0} \\ \bm{\ast}_{\kappa_{L}}  \end{pmatrix},\ldots,\begin{pmatrix} \bm{0} \\ \bm{\ast}_{\kappa_{L}} \end{pmatrix}}_{\kappa_{L}\text{-many}} 
}
\end{equation}
where $\bm{\ast}_{\kappa_l}$ with $l\in\{1,\ldots,L\}$ are vectors of length $\kappa_l$ and $\bm{0}$ are vectors of suitable dimension containing zeros. The first $\kappa_{b_1}$ eigenvectors have (at least) $M-\kappa_{b_1}$ zero-components, the following $\kappa_{b_2}$ eigenvectors have (at least) $M-\kappa_{b_2}$ zero-components and so on. The eigenvectors of $\Sigmabtildehat$ follow the same shape however they are slightly perturbed due to $\Epsilonb$ and distorted by the noise $\Etab_N$.\\
PLA for discarding, say, $\Beta$ blocks $\Sigmab_{b_1},\ldots,\Sigmab_{b_\Beta}$ with $\Sigmab_{b_\beta} \cind \Sigmab_{b_l}$ $\forall l \neq \beta$ for $\beta\in\{1,\ldots,\Beta\}$ is then given by the following algorithm provided by \cite{BD20}.
\begin{Algorithm}[\textbf{PLA}]\label{alg:PLAcov}
Discard the variables corresponding to $\Sigmab_{b_1},\ldots,\Sigmab_{b_\Beta}$ according to PLA proceeds as follows:
\begin{enumerate}
\item[1.] Check if the eigenvectors of $\Sigmab$ satisfy the required structure in (\ref{eq:blockeigenstructure}) to discard $\Sigmab_{b_1} , \ldots , \Sigmab_{b_\Beta}$.

\item[2.] Decide if $\Sigmab_{b_1} , \ldots , \Sigmab_{b_\Beta}$ are relevant according to the explained variance of the realisations $\{\xb_d\}$ of their contained random variables $\{X_d\}$ by calculating (\ref{eq:ExpVarPLA1}) (or (\ref{eq:ExpVarPCA})).

\item[3.] Discard $\Sigmab_{b_1} , \ldots , \Sigmab_{b_\Beta}$.
\end{enumerate}
\end{Algorithm}
\section{Issues When Using the Covariance Matrix}
\label{s:CovPLA}
While \cite{BD20} propose to check if the absolute elements of the eigenvectors of the covariance matrix are below a threshold $\tau$, we complement their results by providing an incentive to consider the usage of the correlation matrix instead. Our contribution is to show that the small elements of the eigenvectors converge towards zero when $\{\Var(X_d)\}_{d\in\bD}$ increase. This is a concern because the variance is not scale invariant hence PLA might yield different results for the same but rescaled data set.  \\
We consider to drop only a single block, say, $\Sigmab_b$ containing $\{X_d\}$ where it is assumed for convenience purposes that $\bD = \{1,\ldots,M^\ast\}$.  Again, we can assume those elements for $\bD$ because we can obtain this structure of the covariance matrix using a permutation matrix. The extension to the general case when discarding several blocks is analogue.


\begin{Theorem} \label{th:derivative}
Let $\bD$ and $\bDelta$ as introduced in section~\ref{s:setup} and $\sigmatildehat_{jj}\equiv\Vartildehat(X_j)$ and let $i\in\{1,\ldots,M\}$. We assume that $\lambdatildehat_\delta\neq0$. For each $d\in\bD$ there exists one $\delta\in\bDelta$ such that
\begin{equation*}
\dfrac{\partial \vert \vtildehat_\delta^{(i)}\vert}{\partial \Vartildehat(X_{d})}  < 0  \;\;\text{ and }\;\;   \dfrac{\partial \vtildehat_\delta^{(d)}}{\partial \Vartildehat(X_{d})} 
> 0 
\end{equation*}
for $i\neq d$ and as long as $|\vtildehat_\delta^{(d)}| \neq 1$ and $| \vtildehat_\delta^{(i)} | \neq 0$.
\end{Theorem}

\begin{proof}
Let $i,m\in\{1,\ldots,M\}$ with $i\neq d$. From the trace 
\begin{equation}
\label{eq:trace} {\rm tr}(\Sigmabtildehat) = \sum_m \lambdatildehat_m = \sum_m \sigmatildehat_{mm} \equiv \sum_m \Vartildehat(X_m)
\end{equation}
we can conclude that $\lambdatilde_j = \sum_m \Vartildehat(X_{m}) - \sum_{m\neq j} \lambdatildehat_j$. If we consider now that $\Vartildehat(X_{j})$ changes by, say, $\mu_j$
$$\Vartildehat(X_{j}) \mapsto \Vartildehat(X_{j}) + \mu_j$$
it holds that
\begin{equation}
\label{eq:lambdachange}
\lambdatildehat_m \mapsto \lambdatildehat_m + p_{mj}\mu_j
\end{equation}
changes as well with $p_{mj}\in[0,1]$ and $\sum_m p_{mj} = 1$ such that (\ref{eq:trace}) is satisfied.\\
Further, from the eigendecomposition in (\ref{eq:eigendec}) we can conclude that $\vbtildehat_\delta= \lambdatildehat_\delta^{-1} \Sigmabtildehat \vbtildehat_\delta$ and hence from (\ref{eq:trace}) that
\begin{align}
& \vtildehat_\delta^{(i)} = \dfrac{\sum_m \sigmatildehat_{i m} \vtildehat_\delta^{(m)}}{ \lambdatildehat_\delta } = \dfrac{\sum_m \sigmatildehat_{i m} \vtildehat_\delta^{(m)} }{\sum_m \sigmatildehat_{mm} - \sum_{m\neq \delta} \lambdatildehat_m }  \label{eq:vtildehat} \\
& \vtildehat_\delta^{(d)} = \dfrac{\sum_m \sigmatildehat_{dm} \vtildehat_\delta^{(m)}}{ \lambdatildehat_\delta } = \dfrac{\sigmatildehat_{dd}\vtildehat_\delta^{(d)} + \sum_{m\neq d} \sigmatildehat_{dm} \vtildehat_\delta^{(m)} }{\sum_m \sigmatildehat_{mm} - \sum_{m\neq \delta} \lambdatildehat_m }  \label{eq:vtildehatd} 
 \end{align}
which can both be considered as a function of $\Vartildehat(X_d) \equiv \sigmatildehat_{dd}$. We can now derive the partial derivatives.\\
{\rm Case I}: starting with the case that $p_{\delta d}\in(0,1]$ i.e. that $p_{\delta d}\neq 0$ and $\sum_{m\neq\delta} p_{md} < 1$, we obtain from (\ref{eq:vtildehat}) due to (\ref{eq:trace}) and (\ref{eq:lambdachange}) that
\begin{align}
& \vtildehat_\delta^{(i)}\left(\Vartildehat(X_d) + \mu_d \right) \nonumber \\
\equiv\;\; & \vtildehat_\delta^{(i)}\left(\sigmatildehat_{dd} + \mu_d \right) \nonumber \\
= \;\;& \dfrac{\sum_m \sigmatildehat_{i m} \vtildehat_\delta^{(m)} }{\mu_d\left(1-\sum_{m\neq\delta} p_{md}\mu_d\right) + \sum_{m} \sigmatildehat_{mm} - \sum_{m\neq \delta} \lambdatildehat_m  }  \nonumber \\
= \;\;& \dfrac{\sum_m \sigmatildehat_{i m} \vtildehat_\delta^{(m)} }{\mu_d\left(1-\sum_{m\neq\delta} p_{md}\mu_d\right) + \lambdatildehat_\delta  } \nonumber \\
= \;\;& \dfrac{\sum_m \sigmatildehat_{im} \vtildehat_\delta^{(m)} }{ p_{\delta d}\mu_d + \lambdatildehat_\delta  } \;. \label{eq:vtildehat2}
 \end{align}
Let $f(\mu_d)\equiv \big| \vtildehat_\delta^{(i)}\left(\Vartildehat(X_d) + \mu_d \right) \big|  - \big|  \vtildehat_\delta^{(i)}\left(\Vartildehat(X_d) \right)\big| $. From (\ref{eq:vtildehat}) and (\ref{eq:vtildehat2}) we see that $\lim_{\mu_d\to0} f(\mu_d) = 0$. Hence, the partial derivative is given by
\begin{multline}
\dfrac{\partial \vert \vtildehat_\delta^{(i)}\vert}{\partial \Vartildehat(X_{d})}  =\lim\limits_{\mu_d\to0} \dfrac{f(\mu_d)}{\mu_d} = \lim\limits_{\mu_d\to0} \dfrac{\partial f(\mu_d) }{\partial \mu_d}    \\
=  - \dfrac{| \sum\limits_{m} \sigmatildehat_{im} \vtildehat_\delta^{(m)}  | \lambdatildehat_\delta  }{|\lambdatildehat_\delta |^3 } \label{eq:derivativeproof}
\end{multline}
where we used L'Hospital's rule in the second step since $\lim_{\mu_d\to0}\mu_d = 0$ as well as $\partial\mu_d/\partial\mu_d = 1$. The final step is an immediate result of the chain rule $\partial/\partial x \cdot 1/|x| = \partial/\partial|x| \cdot 1/|x| \cdot \partial|x|/\partial x$. The result follows since $\Sigmabtildehat$ is positive semi-definite by construction hence $\lambdatildehat_\delta > 0$ because $\lambdatildehat_\delta\neq0$ is assumed.\\
To obtain the second result we conclude from $\Vert\vbtildehat_\delta\Vert_2^2 = 1$ that
\begin{equation}
\label{eq:normalized}
\vtildehat_\delta^{(d)} =\sqrt{ 1 - \sum_{m\neq d} ( \vtildehat_\delta^{(m)})^2} \; .
\end{equation}
If now $|\vtildehat_\delta^{(i)}| \neq 0$ decreases which is the case if $\Vartildehat(X_{d})$ increases since $\partial \vert \vtildehat_\delta^{(i)}\vert / \partial \Vartildehat(X_{d}) <0$, then $|\vtildehat_\delta^{(d)} |$ increases due to (\ref{eq:normalized}). Hence $\partial \vert \vtildehat_\delta^{(d)}\vert / \partial \Vartildehat(X_{d}) > 0 $. \\
{\rm Case II}: when $p_{\delta d}=0$ we obtain from (\ref{eq:vtildehatd}) that 
\begin{equation} \label{eq:vtildehatd2}
\vtildehat_\delta^{(d)}  = \dfrac{(\Vartildehat(X_d) + \mu_d) \vtildehat_\delta^{(d)} +  \sum_{m\neq d} \sigmatildehat_{dm} \vtildehat_\delta^{(m)}  }{\lambdatildehat_\delta}
\end{equation}
since $\lambdatildehat_\delta$ does not change in $\Vartildehat(X_d) $. Analogue to {\rm case I}, let $g(\mu_d)\equiv \big| \vtildehat_\delta^{(d)}\left(\Vartildehat(X_d) + \mu_d \right) \big|  - \big|  \vtildehat_\delta^{(d)}\left(\Vartildehat(X_d) \right)\big| $. From (\ref{eq:vtildehatd}) and (\ref{eq:vtildehatd2}) we see that $\lim_{\mu_d\to0} g(\mu_d) = 0$. Hence, the partial derivative is given by
\begin{equation*}
\dfrac{\partial \vert \vtildehat_\delta^{(d)}\vert}{\partial \Vartildehat(X_{d})}  =\lim\limits_{\mu_d\to0} \dfrac{g(\mu_d)}{\mu_d} = \lim\limits_{\mu_d\to0} \dfrac{\partial g(\mu_d) }{\partial \mu_d}  \\
=   \dfrac{| \vtildehat_\delta^{(d)}  |   }{|\lambdatildehat_\delta | } > 0
\end{equation*}
following the same arguments as in (\ref{eq:derivativeproof}) and due to the structure in (\ref{eq:blockeigenstructure}). We obtain that $\partial \vert \vtildehat_\delta^{(i)}\vert / \partial \Vartildehat(X_{d})  < 0$ analogue to {\rm case I} when solving (\ref{eq:normalized}) for $\vtildehat_\delta^{(i)}$ instead of $\vtildehat_\delta^{(d)}$ and follow the arguments above however reversed by considering first that $|\vtildehat_\delta^{(d)}|\neq 1$ increases.
\end{proof}
\noindent
The intuition behind Theorem~\ref{th:derivative} is that $\Var(X_d)$ is present in both, the numerator and denominator of $\vtildehat_\delta^{(d)}$ while $\Var(X_d)$ only enters the denominator of $\vtildehat_\delta^{(i)}$ for $i\neq d$ via $\lambdatildehat_d$. Strictly speaking, a change of $\Var(X_d)$ enters also the numerator of $\vtildehat_\delta^{(i)}$ due to $\Covtildehat(X_i,X_d) \equiv\sigmatildehat_{id}$. Consider to increase or decrease $X$ by $c\cdot X$, for $c\in\R$ being a constant, as it might occur when changing scales. Then $\Covtildehat(X_i,c\cdot X_d) = c \cdot \Covtildehat(X_i,X_d)$. However, since $\Vartildehat(c\cdot X_d) = c^2\cdot \Vartildehat(X_d)$ we can simply shorten the fraction to get rid of $c$ in the numerator. When considering a change in a variable as a sum, as we did in the proof of Theorem~\ref{th:derivative}, the change will not be present in the numerator since $\Covtildehat(X_i,X_d + c) = \Covtildehat(X_i,X_d)$. Further, the assumption that $\lambdatildehat_\delta\neq0$ is reasonable since this case barely occurs for a covariance matrix and is less strict than assuming positive definiteness. 
\begin{Corollary} \label{co:derivative}
Let $\bD$ and $\bDelta$ as introduced in section~\ref{s:setup} and $\sigmatildehat_{jj}\equiv\Vartildehat(X_j)$. We assume that $\lambdatildehat_\delta\neq0$. For each $d\in\bD$ there exists one $\delta\in\bDelta$ such that for all $d^c\notin\bD$
$$ \dfrac{\partial \vert \vtildehat_\delta^{(d^c)}\vert}{\partial \Vartildehat(X_{d})}  < 0 \;.$$
\end{Corollary}
\noindent
Corollary~\ref{co:derivative} shows the issue when using the covariance matrix for PLA. Since the covariance matrix is not scale invariant, the respective eigenvectors are neither: if the variance decreases, the elements that we check to lie under a certain threshold increase. Further they decrease, in fact converge towards zero $\lim_{\Vartildehat(X_d)\to\infty}\vtildehat_\delta^{(d^c)} \to 0$, if the variance increases.
\begin{Corollary} \label{co:derivative2}
Let $\bD$ and $\bDelta$ as introduced in section~\ref{s:setup} and $\sigmatildehat_{jj}\equiv\Vartildehat(X_j)$. We assume that $\lambdatildehat_\delta\neq0$. For each $d\in\bD$ there exists one $\delta\in\bDelta$ such that for all $i\neq\delta$
$$ \dfrac{\partial \vert \vtildehat_i^{(d)}\vert}{\partial \Vartildehat(X_{d})}  < 0 \;.$$
\end{Corollary}
\begin{proof}
The result follows from Theorem~\ref{th:derivative} and since the rows of $\Vbtildehat$ are normed to 1.
\end{proof}
\noindent
When introducing PLA, \cite{BD20} also considered to check not (only) the rows of the eigenvectors but (also) the columns. According to Corollary~\ref{co:derivative2} we face the same problem when following the latter approach. 
\section{PLA Using the Correlation Matrix}
\label{s:CorrPLA}
In this section, we introduce PLA based on the correlation matrix and address a concern regarding the eigenvalues.\\
Since, both the correlation matrix $\Rhob$ is invariant to linear changes \cite{JC16} and the eigenvectors have the same shape (\ref{eq:blockeigenstructure}) as the eigenvectors of $\Sigmab$, PLA based on $\Rhob$ is a natural choice. However, there is a downside when it comes down to calculate the explained variance. For simplicity, we consider to use (\ref{eq:ExpVarPCA}). Note that the same issues hold for (\ref{eq:ExpVarPLA1}) as well. We denote the block structure of $\varepsilon$-uncorrelated random variables analogue to (\ref{eq:blockcovariancematrix}) by
\begin{equation*}
\underbrace{\Rhobtilde_1}_{\kappa_1\times\kappa_1}  \epsind   \ldots \epsind \underbrace{\Rhobtilde_L}_{\kappa_L\times\kappa_L} \;.
\end{equation*}
Since ${\rm tr}(\Rhob) = \sum_i \omega_i = \sum_i \rho_{ii} = M$ and  ${\rm tr}(\Rhob_l) = \sum_{\delta_l} \omega_{\delta_l} = \sum_{m=M_{l-1}}^{M_l} \rho_{mm} = \kappa_l$, where $\delta_l$ indexes the eigenvalues corresponding to the eigenvectors linked to the random variables contained in $\Rhob_l$, the explained variance for any block $\Rhob_l$ is given by
\begin{equation*}
\left(  \sum\limits_{i} \omega_i  \right)^{-1} \left( \sum\limits_{\delta_l} \omega_{\delta_l} \right)  = \dfrac{\kappa_l}{M} \; . 
\end{equation*}
Hence, the explained variance for any block $\Rhobtildehat_l$ is approximately given by $\kappa_l/M$, since $\Epsilonb$ and $\Etab_N$ are sparse. This means that we loose the information provided by $\Lambdabtildehat$ to evaluate the importance of each block since we standardized each variable to unit variance. Therefore, we propose to use the eigenvectors of $\Rhobtildehat$ to search the blocks of concern and to use the eigenvalues of $\Sigmabtildehat$ to decide whether to discard or not. However, that $\Sigmabtildehat$ is not scale invariant is the price we pay when coming back to $\Sigmabtildehat$ to calculate the explained variance. Since this is a well known concern in classic PCA, we refer to \cite{FL97} and \cite{JO02} for elaborate explanations. 
\begin{Algorithm}[\textbf{PLA based on the correlation matrix}]\label{alg:PLAcor}
Discard the variables corresponding to $\Rhob_{b_1},\ldots,\Rhob_{b_\Beta}$ according to PLA based on the correlation matrix proceeds as follows:
\begin{enumerate}
\item[1.] Check if the eigenvectors of $\Rhob$ satisfy the required structure in (\ref{eq:blockeigenstructure}) to discard $\Rhob_{b_1} , \ldots , \Rhob_{b_\Beta}$.

\item[2.] Decide if $\Rhob_{b_1} , \ldots , \Rhob_{b_\Beta}$ are relevant according to the explained variance of the realisations $\{\xb_d\}$ of their contained random variables $\{X_d\}$ by calculating (\ref{eq:ExpVarPLA1}) (or (\ref{eq:ExpVarPCA})).

\item[3.] Discard $\Rhob_{b_1} , \ldots , \Rhob_{b_\Beta}$.
\end{enumerate}
\end{Algorithm}
\section{Rescaled Eigenvectors}
\label{s:rescaledeigenvectors}
A minor addition to PLA is rescaling the eigenvectors. In this section we briefly cover this change and provide a modified algorithm. \\
Originally coming from PCA, the idea is to rescale the eigenvectors so the maximum value equals one \cite{JO02} which is easily done by dividing by the largest element. Hence, we do not check the elements of $\vbtildehat_j$ but rather the elements of
\begin{equation}\label{eq:rescaledeigenvectors}
\ubtildehat_j/{\rm arg\; max}_i |\ubtildehat_j^{(i)}| \; .
\end{equation}
This modifies PLA to a more standardised procedure. We provide the algorithm when using rescaled eigenvectors of the correlation matrix. The change of the algorithm based on the covariance matrix is analogue.

\begin{Algorithm}[\textbf{PLA based on rescaled eigenvectors of the correlation matrix}]\label{alg:PLAcorscale}
Discard the variables corresponding to $\Rhob_{b_1},\ldots,\Rhob_{b_\Beta}$ according to PLA based on rescaled eigenvectors of the correlation matrix proceeds as follows:
\begin{enumerate}
\item[1.] Check if the rescaled eigenvectors (\ref{eq:rescaledeigenvectors}) of $\Rhob$ satisfy the required structure in (\ref{eq:blockeigenstructure}) to discard $\Rhob_{b_1} , \ldots , \Rhob_{b_\Beta}$.

\item[2.] Decide if $\Rhob_{b_1} , \ldots , \Rhob_{b_A}$ are relevant according to the explained variance of the realisations $\{\xb_d\}$ of their contained random variables $\{X_d\}$ by calculating (\ref{eq:ExpVarPLA1}) (or (\ref{eq:ExpVarPCA})).

\item[3.] Discard $\Rhob_{b_1} , \ldots , \Rhob_{b_\Beta}$.
\end{enumerate}
\end{Algorithm}
\section{Simulation Study}
\label{s:simulation}
We conduct a simulation study in this section to evaluate the performance of PLA based on the rescaled eigenvectors of the correlation matrix for different threshold values. There is a concern regarding the simulation due to the perturbations $\Epsilonb$ and $\Etab$ which we discuss as well. \\
Choosing the optimal cut-off value $\tau$ is crucial for PLA, however finding such a value theoretically is rather difficult due to the fuzziness of algorithm step 2\cite{BD20}. Hence, we conducted a simulation study. We simulated the case when dropping $k\in\{1,\ldots,5\}$ uncorrelated blocks of dimension $1\times1$, i.e. single variables, and the case when dropping a single uncorrelated $\kappa\times\kappa$ block with $\kappa\in\{2,\ldots,6\}$. The population $\Xb$ consisting of $M$ variables with $\num{100000}$ realisations was simulated $S=\num{10000}$ times. Then, for each $S$ a sample $\xb$ of size $N \in\{ \num{5000},\num{10000}\}$ has been drawn and we conducted PLA according to Algorithm~\ref{alg:PLAcorscale} for $\tau\in\{0.4,0.5,$ $0.6,0.7\}$ and $\tau\in\{0.6,0.7,0.8,0.9\}$ for the single variable case and for the block case respectively. We considered cut-off values used in published studies as an orientation \cite{PNJS03}.\footnote{We considered cut-off values ranging from $0.1$ to $0.9$ during research. However, we only present thresholds that are suitable for practice. Further, the tails of the tables sufficiently indicate the decrease in performance for wider or tighter thresholds. This makes illustrating more extreme cut-off values dispensable.} The resulting type {\rm I} error probabilities are calculated as the share of iterations where PLA did not lead to a consideration of a drop.\\
The concern when using a simulation study to find optimal thresholds is described in \refBauer. The theorem provides an intuition of the possible magnitude of the perturbations of $\Sigmab$ that results in a drop. The result for the perturbations of $\Rhob$ is analogue. For completion, we restate the theorem as well as the proof.
\begin{Theorem}\label{th:Kahan}
Denote $\tilde{\lambda}_0 = \lambda_0  \equiv \infty$ and $\tilde{\lambda}_{M+1}  = \lambda_{M+1} \equiv -\infty$. For $j\in\{1,\ldots,M\}$ it holds that
$$ \dfrac{2^{3/2} \|\Epsilonb + \Etab_N \|_F }{\min(\lambda_{j-1} - \lambda_j,\lambda_j - \lambda_{j+1}) } < \tau\;\Rightarrow\;  \|\epsilonb_j + \etab_\jN \|_\infty < \tau \; . $$
\end{Theorem}
\begin{proof}
From \refDavis$\;$ we can conclude that $\|\epsilonb_j + \etab_\jN \|_2 \leq  2^{3/2} \|\Epsilonb + \Etab_N \|_F / \min(\lambda_{j-1} - \lambda_j,\lambda_j - \lambda_{j+1})$ which yields our desired result since $\|\epsilonb_j + \etab_\jN \|_\infty \leq \|\epsilonb_j + \etab_\jN\|_2$. 
\end{proof}\noindent
Hence, discarding depends on the size of $\Epsilonb$ and $\Etab_N$ which enlarges the amount of parameters that have to be simulated. In this work however we simulate the special case when $\Epsilonb = \bm{0}$ and focus on the influence of the sample noise reflected by the sample size $N$. For completion, we shall emphasise that Theorem~\ref{th:Kahan} is not always feasible for $\Rhob$ without assuming that the eigenvalues are distinct since the eigenvalues for $\Rhob$ are more close in general. Nonetheless, the intuition behind the theorem is valid.\\
In Table~\ref{tab:alphavariable} in Appendix~\ref{a:Tables} we see that the type {\rm I} error for discarding single uncorrelated variables is smaller $0.05$ for most cases when $\tau\in\{0.5,0.6\}$. Of course, the error probability decreases when the thresholds increases. However, since one should expect that the type {\rm II} error increases with larger thresholds, we recommend to use the smallest cut-off yielding sufficient results hence $\tau \leq 0.6$. $\tau \equiv \tau(N,M,k)$ is hereby a function of sample size, number of variables and number of uncorrelated variables and can be adjusted according to those values. In an analogue manner according to Table~\ref{tab:alphablock}, the type {\rm I} errors for single uncorrelated blocks performs well for $\tau\leq 0.8$ where $\tau \equiv \tau(N,M,\kappa)$ is a function of sample size, number of variables and the dimension of the uncorrelated block. We shall emphasize however that the choice of thresholds depends on the data as well as on the purpose of statistical analysis. Hence, choosing even smaller or wider cut-off values might be reasonable if larger type {\rm I} or type {\rm II} errors are tolerable.
\section{Concluding Remarks}
\label{s:conclusion}
We propose to use both, the covariance and the correlation matrix to conduct PLA. This is because the covariance matrix is not scale invariant which may result in different outcomes of PLA. Hence, we recommend to use Algorithm~\ref{alg:PLAcor} or Algorithm~\ref{alg:PLAcorscale} instead. For the latter one, an orientation is to use a threshold $\tau\leq 0.6$ for the case of single uncorrelated variables and $\tau\leq0.8$ for a block of uncorrelated variables.

\bibliographystyle{IEEEtranS}
\bibliography{bib}

\onecolumn
\appendix
\section{Threshold Values}
\label{a:Tables}\noindent
We provide the type {\rm I}  error rates for $k\in\{1,\ldots,5\}$ single uncorrelated variables and for an uncorrelated $\kappa\times\kappa$ block with $\kappa\in\{2,\ldots,6\}$ respectively. As specified in Section~\ref{s:simulation}, the error probabilities are calculated as the share of iterations where PLA did not lead to a consideration of a drop.

\begin{ThreePartTable}
\begin{TableNotes}
\footnotesize
\item \textit{Notes:} the type {\rm I} error is computed as the share of iterations where the $k$ variables have not been discarded.
\end{TableNotes}
\begin{longtable}{rlllllrllllll}
\caption{\label{tab:alphavariable}Type {\rm I} error for $k\in\{1,2,3,4,5\}$ uncorrelated variables with sample size $N\in\{ \num{5000}, \num{10000} \}$ and threshold $\tau\in\{0.4,0.5,0.6,0.7\}$}\\

\toprule
\multicolumn{6}{c}{$N=\num{5000}$}&\multicolumn{6}{c}{$N=\num{10000}$}   \\
\cmidrule(l{3pt}r{3pt}){1-6} \cmidrule(l{3pt}r{3pt}){7-12}
 $M$  & $k$ & $\tau=0.4$ & $\tau=0.5$ & $\tau=0.6$ & $\tau=0.7$ & $M$  & $k$  & $\tau=0.4$ & $\tau=0.5$ & $\tau=0.6$ & $\tau=0.7$\\
\cmidrule(l{3pt}r{3pt}){1-2} \cmidrule(l{3pt}r{3pt}){3-6} \cmidrule(l{3pt}r{3pt}){7-8} \cmidrule(l{3pt}r{3pt}){9-12}
\endfirsthead
\caption[]{Type {\rm I} error for $k\in\{1,2,3,4,5\}$ uncorrelated variables with sample size $N\in\{ \num{5000}, \num{10000} \}$ and threshold $\tau\in\{0.4,0.5,0.6,0.7\}$ \textit{(continued)}}\\
\toprule
\multicolumn{6}{c}{$N=\num{5000}$}&\multicolumn{6}{c}{$N=\num{10000}$}   \\
\cmidrule(l{3pt}r{3pt}){1-6} \cmidrule(l{3pt}r{3pt}){7-12}
 $M$  & $k$ & $\tau=0.4$ & $\tau=0.5$ & $\tau=0.6$ & $\tau=0.7$ & $M$  & $k$  & $\tau=0.4$ & $\tau=0.5$ & $\tau=0.6$ & $\tau=0.7$\\
\cmidrule(l{3pt}r{3pt}){1-2} \cmidrule(l{3pt}r{3pt}){3-6} \cmidrule(l{3pt}r{3pt}){7-8} \cmidrule(l{3pt}r{3pt}){9-12}
\endhead
\
\endfoot
\bottomrule
\insertTableNotes
\endlastfoot
        20   &   1 &  0.0434 &  0.0079 &  0.0014 &  0.0000 & 
        20   &   1 &  0.0245 &  0.0041 &  0.0000 &  0.0000 \\      
        40   &   1 &  0.0419 &  0.0053 &  0.0003 &  0.0000 & 
        40   &   1 &  0.0196 &  0.0014 &  0.0000 &  0.0000 \\        
        60   &   1 &  0.0489 &  0.0058 &  0.0005 &  0.0000 & 
        60   &   1 &  0.0174 &  0.0023 &  0.0001 &  0.0001 \\        
        80   &   1 &  0.0576 &  0.0065 &  0.0004 &  0.0001 & 
        80   &   1 &  0.0159 &  0.0010 &  0.0001 &  0.0000 \\        
       100  &   1 &  0.0688 &  0.0080 &  0.0004 &  0.0000 & 
       100  &   1 &  0.0192 &  0.0012 &  0.0000 &  0.0000 \\       
       120  &   1 &  0.0822 &  0.0090 &  0.0002 &  0.0001 & 
       120  &   1 &  0.0205 &  0.0012 &  0.0001 &  0.0000 \\       
       140  &   1 &  0.0911 &  0.0119 &  0.0006 &  0.0002 & 
       140  &   1 &  0.0226 &  0.0008 &  0.0000 &  0.0000 \\      
       160  &   1 &  0.1070 &  0.0131 &  0.0013 &  0.0001 & 
       160  &   1 &  0.0250 &  0.0010 &  0.0001 &  0.0000 \\       
       180  &   1 &  0.1226 &  0.0140 &  0.0016 &  0.0000 & 
       180  &   1 &  0.0281 &  0.0019 &  0.0000 &  0.0000 \\       
       200  &   1 &  0.1462 &  0.0176 &  0.0011 &  0.0000 & 
       200  &   1 &  0.0285 &  0.0025 &  0.0002 &  0.0000 \\    
                          \addlinespace   
        20   &   2 &  0.1055 &  0.0364 &  0.0103 &  0.0011 & 
        20   &   2 &  0.0681 &  0.0232 &  0.0065 &  0.0009 \\       
        40   &   2 &  0.1293 &  0.0297 &  0.0054 &  0.0013 & 
        40   &   2 &  0.0692 &  0.0151 &  0.0015 &  0.0000 \\        
        60   &   2 &  0.1411 &  0.0291 &  0.0042 &  0.0005 & 
        60   &   2 &  0.0660 &  0.0105 &  0.0007 &  0.0000 \\        
        80   &   2 &  0.1588 &  0.0299 &  0.0043 &  0.0005 & 
        80   &   2 &  0.0611 &  0.0087 &  0.0014 &  0.0002 \\        
       100  &   2 &  0.1743 &  0.0318 &  0.0056 &  0.0002 & 
       100  &   2 &  0.0636 &  0.0069 &  0.0007 &  0.0000 \\      
       120  &   2 &  0.2014 &  0.0337 &  0.0039 &  0.0003 & 
       120  &   2 &  0.0651 &  0.0086 &  0.0005 &  0.0002 \\       
       140  &   2 &  0.2281 &  0.0385 &  0.0047 &  0.0007 & 
       140  &   2 &  0.0721 &  0.0080 &  0.0005 &  0.0000 \\       
       160  &   2 &  0.2434 &  0.0438 &  0.0042 &  0.0003 & 
       160  &   2 &  0.0754 &  0.0088 &  0.0004 &  0.0000 \\       
       180  &   2 &  0.2650 &  0.0507 &  0.0059 &  0.0005 & 
       180  &   2 &  0.0802 &  0.0069 &  0.0006 &  0.0000 \\       
       200  &   2 &  0.2950 &  0.0507 &  0.0077 &  0.0003 & 
       200  &   2 &  0.0861 &  0.0095 &  0.0004 &  0.0001 \\     
                          \addlinespace  
        20   &   3 &  0.1470 &  0.0546 &  0.0159 &  0.0037 & 
        20   &   3 &  0.0970 &  0.0397 &  0.0093 &  0.0022 \\        
        40   &   3 &  0.1978 &  0.0511 &  0.0107 &  0.0017 & 
        40   &   3 &  0.1085 &  0.0270 &  0.0038 &  0.0006 \\        
        60   &   3 &  0.2258 &  0.0508 &  0.0093 &  0.0010 & 
        60   &   3 &  0.1011 &  0.0178 &  0.0025 &  0.0006 \\        
        80   &   3 &  0.2513 &  0.0513 &  0.0078 &  0.0008 & 
        80   &   3 &  0.1071 &  0.0156 &  0.0026 &  0.0000 \\        
       100  &   3 &  0.2749 &  0.0532 &  0.0081 &  0.0007 & 
       100  &   3 &  0.1079 &  0.0157 &  0.0020 &  0.0003 \\       
       120  &   3 &  0.3076 &  0.0587 &  0.0079 &  0.0013 & 
       120  &   3 &  0.1130 &  0.0131 &  0.0013 &  0.0001 \\       
       140  &   3 &  0.3303 &  0.0626 &  0.0080 &  0.0008 & 
       140  &   3 &  0.1136 &  0.0139 &  0.0010 &  0.0003 \\       
       160  &   3 &  0.3722 &  0.0687 &  0.0093 &  0.0007 & 
       160  &   3 &  0.1230 &  0.0156 &  0.0017 &  0.0002 \\              
       180  &   3 &  0.3964 &  0.0771 &  0.0089 &  0.0008 & 
       180  &   3 &  0.1348 &  0.0155 &  0.0015 &  0.0001 \\       
       200  &   3 &  0.4315 &  0.0845 &  0.0094 &  0.0010 & 
       200  &   3 &  0.1415 &  0.0136 &  0.0011 &  0.0001 \\   
                          \addlinespace    
        20   &   4 &  0.1981 &  0.0807 &  0.0265 &  0.0074 & 
        20   &   4 &  0.1251 &  0.0501 &  0.0169 &  0.0037 \\        
        40   &   4 &  0.2566 &  0.0745 &  0.0153 &  0.0032 & 
        40   &   4 &  0.1455 &  0.0361 &  0.0056 &  0.0015 \\               
        60   &   4 &  0.2957 &  0.0678 &  0.0117 &  0.0029 & 
        60   &   4 &  0.1472 &  0.0299 &  0.0051 &  0.0006 \\        
        80   &   4 &  0.3213 &  0.0777 &  0.0118 &  0.0022 & 
        80   &   4 &  0.1464 &  0.0254 &  0.0028 &  0.0004 \\        
       100  &   4 &  0.3680 &  0.0777 &  0.0109 &  0.0027 & 
       100  &   4 &  0.1543 &  0.0228 &  0.0027 &  0.0002 \\      
       120  &   4 &  0.3953 &  0.0824 &  0.0107 &  0.0014 & 
       120  &   4 &  0.1552 &  0.0224 &  0.0019 &  0.0002 \\       
       140  &   4 &  0.4302 &  0.0924 &  0.0113 &  0.0018 & 
       140  &   4 &  0.1672 &  0.0216 &  0.0026 &  0.0008 \\       
       160  &   4 &  0.4702 &  0.0949 &  0.0112 &  0.0002 & 
       160  &   4 &  0.1736 &  0.0201 &  0.0018 &  0.0000 \\       
       180  &   4 &  0.4975 &  0.1044 &  0.0133 &  0.0018 & 
       180  &   4 &  0.1789 &  0.0218 &  0.0025 &  0.0000 \\       
       200  &   4 &  0.5409 &  0.1168 &  0.0157 &  0.0012 & 
       200  &   4 &  0.1949 &  0.0186 &  0.0017 &  0.0003 \\     
                          \addlinespace  
        20   &   5 &  0.2170 &  0.0977 &  0.0373 &  0.0103 & 
        20   &   5 &  0.1451 &  0.0667 &  0.0214 &  0.0055 \\        
        40   &   5 &  0.3107 &  0.0961 &  0.0237 &  0.0050 & 
        40   &   5 &  0.1784 &  0.0547 &  0.0109 &  0.0019 \\        
        60   &   5 &  0.3588 &  0.0956 &  0.0162 &  0.0033 & 
        60   &   5 &  0.1884 &  0.0375 &  0.0077 &  0.0007 \\        
        80   &   5 &  0.4057 &  0.1005 &  0.0163 &  0.0023 & 
        80   &   5 &  0.1881 &  0.0347 &  0.0033 &  0.0005 \\        
       100  &   5 &  0.4401 &  0.1035 &  0.0145 &  0.0019 & 
       100  &   5 &  0.1902 &  0.0298 &  0.0036 &  0.0001 \\       
       120  &   5 &  0.4812 &  0.1071 &  0.0161 &  0.0018 & 
       120  &   5 &  0.2008 &  0.0267 &  0.0033 &  0.0002 \\       
       140  &   5 &  0.5248 &  0.1117 &  0.0153 &  0.0017 & 
       140  &   5 &  0.2016 &  0.0243 &  0.0031 &  0.0004 \\       
       160  &   5 &  0.5628 &  0.1234 &  0.0142 &  0.0014 & 
       160  &   5 &  0.2142 &  0.0263 &  0.0028 &  0.0001 \\       
       180  &   5 &  0.5975 &  0.1328 &  0.0170 &  0.0013 & 
       180  &   5 &  0.2240 &  0.0254 &  0.0031 &  0.0001 \\      
       200  &   5 &  0.6382 &  0.1477 &  0.0188 &  0.0017 & 
       200  &   5 &  0.2425 &  0.0276 &  0.0016 &  0.0004\\*

\end{longtable}
\end{ThreePartTable}

\begin{ThreePartTable}
\begin{TableNotes}
\footnotesize
\item \textit{Notes:} the type {\rm I} error is computed as the share of iterations where the block containing $\kappa$ variables has not been discarded.
\end{TableNotes}
\begin{longtable}{rllllllllrllllllll}
\caption{\label{tab:alphablock}Type {\rm I} error for a single uncorrelated $\kappa\times\kappa$ block with $\kappa\in\{2,3,4,5,6\}$, sample size $N\in\{ \num{5000}, \num{10000} \}$ and threshold $\tau\in\{0.6,0.7,0.8,0.9\}$}\\

\toprule
\multicolumn{6}{c}{$N=\num{5000}$}&\multicolumn{6}{c}{$N=\num{10000}$}   \\
\cmidrule(l{3pt}r{3pt}){1-6} \cmidrule(l{3pt}r{3pt}){7-12}
 $M$  & $\kappa$ & $\tau=0.6$ & $\tau=0.7$ & $\tau=0.8$ & $\tau=0.9$ & $M$  & $\kappa$  & $\tau=0.6$ & $\tau=0.7$ & $\tau=0.8$ & $\tau=0.9$\\
\cmidrule(l{3pt}r{3pt}){1-2} \cmidrule(l{3pt}r{3pt}){3-6} \cmidrule(l{3pt}r{3pt}){7-8} \cmidrule(l{3pt}r{3pt}){9-12}
\endfirsthead
\caption[]{Type {\rm I} error for a single uncorrelated $\kappa\times\kappa$ block with $\kappa\in\{2,3,4,5,6\}$, sample size $N\in\{ \num{5000}, \num{10000} \}$ and threshold $\tau\in\{0.6,0.7,0.8,0.9\}$ \textit{(continued)}}\\
\toprule
\multicolumn{6}{c}{$N=\num{5000}$}&\multicolumn{6}{c}{$N=\num{10000}$}   \\
\cmidrule(l{3pt}r{3pt}){1-6} \cmidrule(l{3pt}r{3pt}){7-12}
 $M$  & $\kappa$ & $\tau=0.6$ & $\tau=0.7$ & $\tau=0.8$ & $\tau=0.9$ & $M$  & $\kappa$  & $\tau=0.6$ & $\tau=0.7$ & $\tau=0.8$ & $\tau=0.9$\\
\cmidrule(l{3pt}r{3pt}){1-2} \cmidrule(l{3pt}r{3pt}){3-6} \cmidrule(l{3pt}r{3pt}){7-8} \cmidrule(l{3pt}r{3pt}){9-12}
\endhead
\
\endfoot
\bottomrule
\insertTableNotes
\endlastfoot
       20  &     2  &  0.0684  &  0.0256  &  0.0084  &  0.0011  &
       20  &     2  &  0.0429  &  0.0178  &  0.0041  &  0.0008 \\       
       40  &     2  &  0.0573  &  0.0165  &  0.0031  &  0.0012  &
       40  &     2  &  0.0249  &  0.0064  &  0.0011  &  0.0004 \\      
       60  &     2  &  0.0540  &  0.0143  &  0.0040  &  0.0009  &
       60  &     2  &  0.0174  &  0.0036  &  0.0003  &  0.0001 \\    
       80  &     2  &  0.0629  &  0.0143  &  0.0028  &  0.0003  &
       80  &     2  &  0.0180  &  0.0021  &  0.0001  &  0.0001 \\      
      100 &     2  &  0.0647  &  0.0149  &  0.0030  &  0.0006  &
      100 &     2  &  0.0144  &  0.0035  &  0.0004  &  0.0001 \\    
      120 &     2  &  0.0746  &  0.0176  &  0.0024  &  0.0003  &
      120 &     2  &  0.0175  &  0.0030  &  0.0004  &  0.0000 \\   
      140 &     2  &  0.0790  &  0.0188  &  0.0036  &  0.0003  &
      140 &     2  &  0.0176  &  0.0015  &  0.0001  &  0.0000 \\
      160 &     2  &  0.0903  &  0.0219  &  0.0034  &  0.0005  &
      160 &     2  &  0.0162  &  0.0025  &  0.0003  &  0.0001 \\
      180 &     2  &  0.0931  &  0.0238  &  0.0037  &  0.0006  &
      180 &     2  &  0.0186  &  0.0027  &  0.0000  &  0.0000 \\
      200 &     2  &  0.1038  &  0.0239  &  0.0035  &  0.0009  &
      200 &     2  &  0.0224  &  0.0032  &  0.0003  &  0.0000 \\
                          \addlinespace  
       20  &     3  &  0.0847  &  0.0325  &  0.0091  &  0.0016  &
       20  &     3  &  0.0556  &  0.0209  &  0.0049  &  0.0011 \\
       40  &     3  &  0.0731  &  0.0213  &  0.0048  &  0.0013  &
       40  &     3  &  0.0341  &  0.0063  &  0.0018  &  0.0004 \\
       60  &     3  &  0.0821  &  0.0254  &  0.0039  &  0.0009  &
       60  &     3  &  0.0275  &  0.0049  &  0.0009  &  0.0003 \\
       80  &     3  &  0.0831  &  0.0215  &  0.0049  &  0.0007  &
       80  &     3  &  0.0249  &  0.0053  &  0.0010  &  0.0000 \\
      100 &     3  &  0.0894  &  0.0213  &  0.0052  &  0.0006  &
      100 &     3  &  0.0250  &  0.0055  &  0.0010  &  0.0003 \\
      120 &     3  &  0.0937  &  0.0217  &  0.0052  &  0.0016  &
      120 &     3  &  0.0254  &  0.0048  &  0.0006  &  0.0001 \\
      140 &     3  &  0.1037  &  0.0236  &  0.0040  &  0.0008  &
      140 &     3  &  0.0270  &  0.0047  &  0.0007  &  0.0001 \\
      160 &     3  &  0.1078  &  0.0254  &  0.0061  &  0.0008  &
      160 &     3  &  0.0263  &  0.0054  &  0.0006  &  0.0001 \\
      180 &     3  &  0.1177  &  0.0264  &  0.0040  &  0.0010  &
      180 &     3  &  0.0259  &  0.0039  &  0.0011  &  0.0002 \\
      200 &     3  &  0.1262  &  0.0283  &  0.0050  &  0.0008  &
      200 &     3  &  0.0251  &  0.0052  &  0.0007  &  0.0001 \\
                                \addlinespace  
       20  &     4  &  0.1409  &  0.0658  &  0.0225  &  0.0073  &
       20  &     4  &  0.0895  &  0.0388  &  0.0137  &  0.0037 \\
       40  &     4  &  0.1337  &  0.0449  &  0.0146  &  0.0030  &
       40  &     4  &  0.0704  &  0.0223  &  0.0045  &  0.0005 \\
       60  &     4  &  0.1329  &  0.0442  &  0.0125  &  0.0036  &
       60  &     4  &  0.0567  &  0.0152  &  0.0036  &  0.0003 \\
       80  &     4  &  0.1459  &  0.0397  &  0.0122  &  0.0018  &
       80  &     4  &  0.0547  &  0.0125  &  0.0026  &  0.0006 \\
      100 &     4  &  0.1595  &  0.0461  &  0.0105  &  0.0023  &
      100 &     4  &  0.0514  &  0.0116  &  0.0017  &  0.0003 \\
      120 &     4  &  0.1622  &  0.0496  &  0.0114  &  0.0023  &
      120 &     4  &  0.0503  &  0.0111  &  0.0013  &  0.0002 \\
      140 &     4  &  0.1780  &  0.0491  &  0.0122  &  0.0030  &
      140 &     4  &  0.0504  &  0.0113  &  0.0024  &  0.0003 \\
      160 &     4  &  0.2021  &  0.0493  &  0.0106  &  0.0020  &
      160 &     4  &  0.0582  &  0.0099  &  0.0013  &  0.0002 \\
      180 &     4  &  0.2069  &  0.0490  &  0.0114  &  0.0017  &
      180 &     4  &  0.0570  &  0.0087  &  0.0014  &  0.0004 \\
      200 &     4  &  0.2065  &  0.0548  &  0.0119  &  0.0021  &
      200 &     4  &  0.0592  &  0.0100  &  0.0013  &  0.0003 \\
                                \addlinespace 
       20  &     5  &  0.1884  &  0.0990  &  0.0443  &  0.0146  &
       20  &     5  &  0.1255  &  0.0698  &  0.0277  &  0.0074 \\
       40  &     5  &  0.1991  &  0.0778  &  0.0244  &  0.0060  &
       40  &     5  &  0.1043  &  0.0342  &  0.0112  &  0.0025 \\
       60  &     5  &  0.1982  &  0.0766  &  0.0191  &  0.0054  &
       60  &     5  &  0.0892  &  0.0280  &  0.0075  &  0.0011 \\
       80  &     5  &  0.2221  &  0.0722  &  0.0199  &  0.0051  &
       80  &     5  &  0.0854  &  0.0258  &  0.0039  &  0.0008 \\
      100 &     5  &  0.2354  &  0.0740  &  0.0193  &  0.0040  &
      100 &     5  &  0.0862  &  0.0215  &  0.0040  &  0.0010 \\
      120 &     5  &  0.2580  &  0.0783  &  0.0233  &  0.0036  &
      120 &     5  &  0.0865  &  0.0207  &  0.0041  &  0.0007 \\
      140 &     5  &  0.2726  &  0.0792  &  0.0209  &  0.0033  &
      140 &     5  &  0.0895  &  0.0187  &  0.0039  &  0.0008 \\
      160 &     5  &  0.2858  &  0.0895  &  0.0180  &  0.0038  &
      160 &     5  &  0.0855  &  0.0188  &  0.0029  &  0.0004 \\
      180 &     5  &  0.2914  &  0.0867  &  0.0222  &  0.0040  &
      180 &     5  &  0.0928  &  0.0201  &  0.0033  &  0.0004 \\
      200 &     5  &  0.3129  &  0.0907  &  0.0192  &  0.0038  &
      200 &     5  &  0.0932  &  0.0191  &  0.0036  &  0.0005 \\
                                \addlinespace 
       20  &     6  &  0.2440  &  0.1413  &  0.0703  &  0.0245  &
       20  &     6  &  0.1708  &  0.0938  &  0.0458  &  0.0178 \\
       40  &     6  &  0.2711  &  0.1194  &  0.0399  &  0.0136  &
       40  &     6  &  0.1564  &  0.0598  &  0.0180  &  0.0056 \\
       60  &     6  &  0.2822  &  0.1073  &  0.0315  &  0.0104  &
       60  &     6  &  0.1399  &  0.0429  &  0.0104  &  0.0031 \\
       80  &     6  &  0.3018  &  0.1082  &  0.0350  &  0.0080  &
       80  &     6  &  0.1334  &  0.0336  &  0.0094  &  0.0014 \\
      100 &     6  &  0.3288  &  0.1137  &  0.0328  &  0.0077  &
      100 &     6  &  0.1340  &  0.0352  &  0.0075  &  0.0012 \\
      120 &     6  &  0.3472  &  0.1197  &  0.0316  &  0.0081  &
      120 &     6  &  0.1363  &  0.0318  &  0.0074  &  0.0008 \\
      140 &     6  &  0.3723  &  0.1210  &  0.0310  &  0.0092  &
      140 &     6  &  0.1388  &  0.0296  &  0.0067  &  0.0007 \\
      160 &     6  &  0.3861  &  0.1281  &  0.0326  &  0.0071  &
      160 &     6  &  0.1337  &  0.0312  &  0.0070  &  0.0009 \\
      180 &     6  &  0.3958  &  0.1404  &  0.0353  &  0.0065  &
      180 &     6  &  0.1399  &  0.0314  &  0.0050  &  0.0004 \\
      200 &     6  &  0.4152  &  0.1389  &  0.0340  &  0.0070  &
      200 &     6  &  0.1405  &  0.0305  &  0.0049  &  0.0008 \\*
\end{longtable}
\end{ThreePartTable}

\end{document}